\documentclass[11pt]{article}
\usepackage[usenames,dvipsnames]{xcolor}
\usepackage{amsmath,amssymb,amsthm,amsfonts}
\usepackage{bbm}
\usepackage[colorlinks,citecolor=blue,bookmarks=true]{hyperref}
\usepackage[letterpaper,margin=1in]{geometry}

\newtheorem{theorem}{Theorem}
\newtheorem{lemma}{Lemma}[section]

\newtheorem{proposition}[lemma]{Proposition}
\newtheorem{fact}[lemma]{Fact}

\theoremstyle{definition}
\newtheorem{definition}[lemma]{Definition}
\newtheorem{remark}[lemma]{Remark}

\theoremstyle{plain}

\newcommand{\E}{\operatorname{{\bf E}}}
\newcommand{\Ex}{\mathop{{\bf E}\/}}
\renewcommand{\Pr}{\operatorname{{\bf Pr}}}
\newcommand{\Prx}{\mathop{{\bf Pr}\/}}

\newcommand{\R}{\mathbbm R}

\newcommand{\wh}[1]{\widehat{#1}}
\newcommand{\sse}{\subseteq}
\newcommand{\bits}{\{-1,1\}}

\newcommand{\bn}{\bits^n}

\newcommand{\isafunc}{:\bn\to\bits}

\renewcommand{\hat}{\wh} 
\newcommand{\hatf}{\wh{f}}

\newcommand{\sumS}{\sum_{S\sse [n]}}

\newcommand{\MAJ}{\mathsf{MAJ}}

\newcommand{\bX}{\mathbf{X}}
\newcommand{\bY}{\mathbf{Y}}
\newcommand{\bZ}{\mathbf{Z}}

\newcommand{\bi}{\boldsymbol{i}}

\newcommand{\bx}{\boldsymbol{x}}

\newcommand{\leaf}{\mathrm{leaf}}

\makeatletter
\newtheorem*{rep@theorem}{\rep@title}
\newcommand{\newreptheorem}[2]{
\newenvironment{rep#1}[1]{
 \def\rep@title{#2 \ref{##1}}
 \begin{rep@theorem}[Restated]\itshape}
 {\end{rep@theorem}}}
\makeatother

\newreptheorem{theorem}{Theorem}
\newreptheorem{lemma}{Lemma}
\newreptheorem{proposition}{Proposition}

\begin{document}

\title{An inequality for the Fourier spectrum of parity decision trees}

\author{Eric Blais\thanks{Part of this research was done while supported by a Simons Postdoctoral Fellowship at MIT.}  \\ University of Waterloo \\ {\tt eric.blais@uwaterloo.ca} 
\and Li-Yang Tan\thanks{Supported by NSF grants CCF-1115703 and CCF-1319788. Part of this research was done while visiting Carnegie Mellon University.}\\ 
Simons Institute, UC Berkeley \\ 
{\tt liyang@cs.columbia.edu} 
\and Andrew Wan\thanks{Part of this research was done while visiting Harvard University and supported by NSF grant CCF-964401.} \\ 
Institute for Defense Analyses\\ 
{\tt atw12@columbia.edu} }

\maketitle

\begin{abstract}
We give a new bound on the sum of the linear Fourier coefficients of a Boolean function in terms of its parity decision tree complexity. This result generalizes an inequality of O'Donnell and Servedio for regular decision trees~\cite{OS08b}. We use this bound to obtain the first 
non-trivial lower bound on the parity decision tree complexity of the
recursive majority function.
\end{abstract}

\section{Introduction}

In this note, we explore connections between two different notions of complexity of
Boolean functions $f\isafunc$: its decision tree complexity, and the sum of its linear
Fourier coefficients.

Decision trees are full binary trees with internal nodes labelled by variables $x_i$ for some 
$i\in [n]$ and with leaves labelled with constants $\ell \in \{-1,1\}$.
A decision tree $D$ is said to compute $f$ if the path from the root to a leaf in $D$ 
defined by $x$ leads to a leaf labelled by $f(x)$ for every $x \in \{-1,1\}^n$.
The \emph{depth} of a decision tree is the maximum number of internal nodes along any
root-to-leaf path, and the \emph{decision tree (depth) complexity} of a function
$f$ is the minimum depth of any decision tree $D$ that computes $f$.

Every Boolean function $f\isafunc$ has a unique representation as a multilinear polynomial 
\[ f(x) = \sumS \hatf(S)\chi_S(x) \]
where $\chi_S(x) := \prod_{i\in S} x_i$
and the numbers 
 $\hat{f}(S) = \E\big[f(\bx)\chi_S(\bx)\big] \in [-1,1]$
are the \emph{Fourier coefficients} of $f$.  
The Fourier coefficients corresponding to singleton sets $S = \{i\}$, $i \in [n]$ are 
called \emph{linear Fourier coefficients}. 
For notational clarity, we will write $\hat{f}(i)$ to denote the linear Fourier coefficient
$\hat{f}(\{i\})$.
As mentioned above, we consider the 
measure of complexity of $f$ determined by the sum $\sum_{i=1}^n \hat{f}(i)$ of its
linear Fourier coefficients.

In~\cite{OS08b} O'Donnell and Servedio established a connection between these two measures
of complexity by establishing the following inequality on the linear Fourier coefficients of a Boolean function computed by a depth-$d$ decision tree: 

\newtheorem*{os2}{O'Donnell--Servedio Inequality}
\begin{os2}
Let $f : \{-1,1\}^n \to \{-1,1\}$ be computable by a decision tree of depth $d$. Then
$
\sum_{i=1}^n \hat{f}(i) \le \sqrt{d}.$  
\end{os2} 

In addition to being a natural statement relating a combinatorial notion of complexity (decision tree complexity) to an analytic one (the sum of linear Fourier coefficients), this inequality is also the crux of the main algorithmic result of \cite{OS08b}, the first algorithm for PAC learning the class of monotone functions to high accuracy from uniformly random labelled examples, running in time polynomial in a reasonable complexity measure of the target function (in this case, its decision tree complexity).  To date this remains our best progress towards the goal of efficiently learning monotone polynomial-sized DNFs, a longstanding open problem in PAC learning \cite{Blum:03tutorial}. 

\subsection{Our main result}

Another notion of complexity of Boolean functions related to decision
trees is their parity decision tree complexity.
\emph{Parity decision trees} (PDTs) are generalizations of decision trees where internal nodes are now labelled
by subsets $S \subseteq [n]$ instead of indices $i \in [n]$, and the edge taken from an internal node is determined by the parity $\bigoplus_{i\in S}x_i$ of the input (instead of the value of the single
value $x_i$ in the case of regular decision trees).
The \emph{parity decision tree (depth) complexity} of a function is the minimum depth of a parity
decision tree that computes $f$. 

Geometrically, parity decision trees correspond to partitions of the hypercube $\{-1,1\}^n$ into
\emph{affine subspaces}, whereas regular decision trees partition the same hypercube into subcubes.
The PDT model of computation has received significant attention in recent years~\cite{montanaro2009communication,zhang2009communication,shaltiel2011dispersers,ben2012affine,tsang2013fourier,cohen2014two,shpilka2014structure,OSTWZ14}, and in particular, there has been much interest in generalizing results that apply to normal decision trees to the more general setting of PDTs (see e.g.~the survey~\cite{zhang2010parity}).  

The parity decision tree complexity of a Boolean function $f$ 
can be much smaller than its regular decision tree complexity.
The parity function over $n$ variables, which can be computed by a trivial parity decision tree of depth $1$ but requires regular decision tree depth $n$, gives the largest possible separation between the two complexity measures.
As a result, many inequalities related to the decision tree complexity do not necessarily hold
with respect to parity decision tree complexity.
In particular, the O'Donnell--Servedio inequality does not imply that any similar inequality must
hold between the sum of linear Fourier coefficients of a Boolean function and its parity decision
tree complexity.  Our main result shows that, nevertheless, such a generalization does hold.

\begin{theorem}
\label{thm:os2}
Let $f : \{-1,1\}^n \to \{-1,1\}$ be computable by a parity decision tree of 
depth $d$.  
Define $\sigma^2 = 4 \Pr[ f(x) = 1 ] \Pr[ f(x) = -1 ]$ to be the variance of $f$. 
Then
$$
\sum_{i=1}^n \hat{f}(i) \le \sqrt{4\ln2\,\sigma^2 d}.\footnote{This result was 
originally circulated in an unpublished
manuscript titled \emph{Discrete isoperimetry via the entropy method} (2013).}
$$
\end{theorem}

The main technical component in the proof of Theorem~\ref{thm:os2} is a
fundamental inequality (presented in Lemma~\ref{lem:parityDT}) concerning
small-depth parity decision trees.
One notable aspect about the proof of this inequality is that it is first
established for a subclass of parity decision trees called
\emph{correlation-free} parity decision trees. 
We then show that every
parity decision tree of depth $d$ can be refined to obtain a 
correlation-free parity decision tree of depth at most $2d$ to 
obtain Lemma~\ref{lem:parityDT}.
See Section~\ref{sec:lemPDT} for the details.

We complete the proof of Theorem~\ref{thm:os2} using an information-theoretic argument. While the proof can also be completed using
analytical arguments and Jensen's inequality, the information-theoretic
argument appears to be required to obtain the sharp bounds in our 
theorem statement. This same argument can also be used in the 
regular decision tree model to sharpen the O'Donnell--Servedio theorem
directly as well.

\subsection{Application: Recursive majority function}

We use Theorem~\ref{thm:os2} to obtain 
the first non-trivial lower bound on the parity decision tree complexity
of the recursive majority function.
The \emph{3-majority function} is the function $\MAJ_3 : \{-1,1\}^3 \to \{-1,1\}$
defined by $\MAJ_3(x) = (-1)^{\mathbf{1}[x_1+x_2+x_3 < 0]}$. For every
$k \ge 2$, the recursive majority function 
$\MAJ_3^{\otimes k} : \{-1,1\}^{3^k} \to \{-1,1\}$ is
defined by setting 
$$
\MAJ_3^{\otimes k}(x) = \MAJ_3\left( \MAJ_3^{\otimes k-1}(x_{\{1,\ldots,3^{k-1}\}}),
\MAJ_3^{\otimes k-1}(x_{\{3^{k-1}+1,\ldots,2\cdot 3^{k-1}\}}),
\MAJ_3^{\otimes k-1}(x_{\{2\cdot 3^{k-1}+1,\ldots,3^k\}}) \right).
$$

The recursive majority function was introduced by Boppana~\cite{SW86}
to determine possible gaps between the deterministic and randomized 
decision tree complexity of Boolean functions. It is easy to verify that
the deterministic decision tree complexity of $\MAJ_3^{\otimes k}$ is $3^k$. 
By contrast, the problem of determining the randomized decision tree complexity
of $\MAJ_3^{\otimes k}$ is much more challenging: following a sequence of 
works on this question~\cite{SW86,JKS03,MNSX11,Leo13,MNS+13}, 
Magniez et al.~\cite{MNS+13} have shown that the minimal depth
$R(\MAJ_3^{\otimes k})$ of any randomized decision tree
that computes the $\MAJ_3^{\otimes k}$ function satisfies
$$
\Omega(2.57143^k) \le R(\MAJ_3^{\otimes k}) \le O(2.64944^k)
$$
but the exact randomized query complexity of the recursive majority function 
is still unknown.

A closely related problem that naturally arises when considering the
recursive majority function is to determine its (deterministic) {parity}
decision tree complexity. 
A standard adversary argument can be used to show that every parity decision 
tree that computes the recursive majority function has depth at least $2^k$. 
Using Theorem~\ref{thm:os2}, we obtain the first lower bound on the
parity decision tree complexity of the recursive majority function that
improves on this trivial lower bound.

\begin{theorem}
\label{thm:3majk}
Every parity decision tree that computes $\MAJ_3^{\otimes k}$ 
has depth $\Omega(2.25^k)$. 
\end{theorem} 

The proof of Theorem~\ref{thm:3majk} is established by computing the linear
Fourier coefficients of the $\MAJ_3$ function directly, using a fundamental
identity on the linear Fourier coefficients of function powers 
(see Fact~\ref{fact:level-1-mult}) to determine the linear Fourier coefficients
of the $\MAJ_3^{\otimes k}$ function, and applying the inequality in
Theorem~\ref{thm:os2}. This approach is quite general, and may be useful
for obtaining lower bounds on the parity decision tree complexity
of other Boolean functions in the future as well.

\section{Preliminaries}\label{sec:prelim}

\subsection{Information theory}

All probabilities and expectations are with respect to the uniform distribution unless otherwise stated.  We use boldface letters (e.g.~$\bX$, $\bx$) to denote random variables. The proof of Theorem~\ref{thm:os2} uses elementary definitions and inequalities from information theory. A more thorough introduction to these tools can be found in~\cite{CoverThomas:91}.

\begin{definition}
The \emph{entropy} of the random variable $\bX$ drawn from the finite sample space $\Omega$ 
according to the probability mass function $p : \Omega \to [0,1]$ is 
$H(\bX) = -\sum_{x \in \Omega} p(x) \log p(x)$.  
The \emph{conditional entropy} of $\bX$ given $\bY$ when they are drawn from the 
joint probability distribution $p : \Omega \times \Omega' \to [0,1]$
is $H(\bX \mid \bY) = - \sum_{x \in \Omega, y\in \Omega'} p(x,y)\log(p(y)/p(x,y))$. 
\end{definition}

\begin{definition}
The \emph{binary entropy function} is the function $h : [0,1] \to \R$  defined by
$h(t) = -t\log_2(t) - (1-t)\log_2(1-t)$. The value $h(t)$ represents the entropy
of a random variable $\bX$ drawn from $\{-1,1\}$ with $\Pr[\bX = 1] = t$.
\end{definition}

\begin{fact}[Data processing inequality]
If $\bX$ and $\bZ$ are conditionally independent given $\bY$, then $H(\bX\mid \bZ) \ge H(\bX\mid \bY).$
\end{fact}

\begin{fact}[Bounds on the binary entropy function]
\label{fact:h-approx}
The binary entropy function $h : [0,1] \to \R$ is bounded above and below by
$
1 - t^2 \le h(\tfrac12 + \tfrac{t}2) \le 1 - \frac{t^2}{2\ln2}.
$
\end{fact}

\subsection{Fourier analysis and function composition}

We assume that the reader is familiar with the Fourier analysis of Boolean functions.
For a complete introduction to the topic, see~\cite{o2007analysis}.

\begin{definition}
The \emph{composition} of 
$f:\{-1,1\}^m\to \{-1,1\}$ and $g:\{-1,1\}^n \to \{-1,1\}$ is
the function $f \circ g : \{-1,1\}^{mn} \to \{-1,1\}$ where
$$
(f \circ g)(x) = f\big( g(x_1,\ldots,x_n),\ldots,g(x_{(m-1)n+1},\ldots,x_{mn}) \big).
$$
For $k \ge 1$, the \emph{$k$th power} of $f : \{-1,1\}^n \to \{-1,1\}$ is the
function $f^{\otimes k} : \{-1,1\}^{n^k} \to \{-1,1\}$ defined recursively by
setting $f^{\otimes 1} = f$ and $f^{\otimes k} = f \circ f^{\otimes k-1}$.
\end{definition}

\begin{remark}
As we can verify directly, the recursive majority function $\MAJ_3^{\otimes k}$ is
the $k$th power of the $\MAJ_3$ function.
\end{remark}

We use the following fact on the linear Fourier coefficients
of composed functions. (See Appendix~\ref{sec:mult} for a proof of this fact.)

\begin{fact}\label{fact:level-1-mult}
For any $f:\{-1,1\}^m\to \{-1,1\}$ and any balanced function 
$g:\{-1,1\}^n \to \{-1,1\}$,
$$
\sum_{k \in [mn]} \widehat{f \circ g}(k) = 
\left(\sum_{i \in [n]} \hat{f}(i) \right) \left( \sum_{j \in [m]} \hat{g}(j) \right).
$$
\end{fact}

\subsection{Parity decision trees}

As mentioned in the introduction, a parity decision tree is a rooted full binary tree where each
internal node is associated with a set $S \subseteq [n]$, the two edges leading to the children
of a node are labelled with $-1$ and $1$, respectively, and each leaf is associated with a value
in $\{-1,1\}$. Each input $x \in \{-1,1\}^n$ defines a path to a unique leaf in a parity
decision tree $T$ by following the edge labelled with $\chi_S(x)$ from a node labelled with $S$.
We say that the tree $T$ \emph{computes} the Boolean function $f : \{-1,1\}^n \to \{-1,1\}$
if each input $x$ defines a path in $T$ to a leaf labelled with $f(x)$. When $T$ computes $f$
and $\ell$ is a leaf of $T$, we write $f(\ell)$ to denote the label of $\ell$.

We can represent each leaf of a parity decision tree $T$ with a vector $\ell \in \{-1,0,1\}^n$
where $\ell_i$ is the expected value of the coordinate $x_i$ over the uniform distribution of all 
inputs $x \in \{-1,1\}^n$ that define a path to the leaf $\ell$ in $T$. We let 
$\leaf_T : \{-1,1\}^n \to \{-1,0,1\}^n$ be the function that returns the vector representation 
of the leaf reached by the path defined in $T$ for every input $x \in \{-1,1\}^n$.

\section{Proof of Theorem~\ref{thm:os2}}
\label{sec:OS} 

The main technical component of the proof of Theorem~\ref{thm:os2} is the following inequality.
 
\begin{lemma}
\label{lem:parityDT}
For any parity decision tree $T$ of depth $d$, 
$
\E_{\ell \in T}\left[ (\sum_{i=1}^n \ell_i)^2 \right] \le 2d.
$
\end{lemma}

We now complete the proof of Theorem~\ref{thm:os2} assuming 
Lemma~\ref{lem:parityDT}. The proof of the lemma then follows in the next
subsection.

\begin{reptheorem}{thm:os2}
Let $f : \{-1,1\}^n \to \{-1,1\}$ be computable by a parity decision tree $T$ of depth $d$. Define
$\sigma^2 = 4 \Pr[ f(x) = 1 ] \Pr[ f(x) = -1 ]$ to be the variance of $f$. Then
$$
\sum_{i=1}^n \hat{f}(i) \le \sqrt{4\ln2\,\sigma^2 d}.
$$
\end{reptheorem}

\begin{proof}
Draw $\bX \in \{-1,1\}^n$ and $\bi \in [n]$ independently and uniformly at random.
Let us first compute the conditional entropy $H(\bX_{\bi} \mid f(\bX))$. Write $\mu = \Pr[ f(\bX) = 1]$.
Then
$$
\Prx_{\bX}[ \bX_i = 1 \mid f(\bX) = 1] 
= \frac{\E[ (\tfrac{1 + \bX_i}2)(\tfrac{1 + f(\bX)}{2})]}{\Pr[f(\bX) = 1]}
= \frac12 + \frac{\hat{f}(i)}{4\mu}
$$
and so
$$
\Prx_{\bX, \bi}[ \bX_{\bi} = 1 \mid f(\bX) = 1] 
= \frac12 + \sum_{i=1}^n \frac{\hat{f}(i)}{4\mu n}.
$$
Similarly,
$$
\Prx_{\bX, \bi}[ \bX_{\bi} = 1 \mid f(\bX) = -1] 
= \frac12 - \sum_{i=1}^n \frac{\hat{f}(i)}{4(1-\mu)n}.
$$
By the definition of conditional entropy and the upper bound in Fact~\ref{fact:h-approx},
\begin{align}
H(\bX_{\bi} \mid f(\bX)) 
&= \mu h\Big(\tfrac12 + \sum_{i=1}^n \frac{\hat{f}(i)}{4\mu n}\Big)
     + (1-\mu) h\Big(\tfrac12 + \sum_{i=1}^n \frac{\hat{f}(i)}{4(1-\mu) n}\Big) \nonumber \\
&\le \mu\left(1 - \frac1{2\ln 2} \left( \frac{\sum_i \hat{f}(i)}{2\mu n} \right)^2 \right) 
     +(1-\mu)\left(1 - \frac1{2\ln 2} \left( \frac{\sum_i \hat{f}(i)}{2(1-\mu)n} \right)^2 \right) \nonumber \\
&= 1 - \frac{\big(\sum_i \hat{f}(i) \big)^2}{8\ln 2 \,\mu(1-\mu) n^2}. \label{eq:entropy-fx}
\end{align}

Since the leaf reached in $T$ by an input $x$ determines $f(x)$, the data processing inequality implies that:
\begin{equation}
\label{eq:entropy-fx-tx}
H(\bX_{\bi} \mid f(\bX)) \ge H(\bX_{\bi} \mid \leaf_T(\bX)).
\end{equation} We also have that
$$
H(\bX_{\bi} \mid \leaf_T(\bX)) = \Ex_{\ell \in T}\Big[ h\Big(\frac12 + \frac{\sum_i \ell_i}{2n}\Big) \Big]
$$
where the expectation is over the distribution defined by the relative mass of each leaf in $T$.
Applying the lower bound in Fact~\ref{fact:h-approx}, we get
\begin{equation}
\label{eq:entropy-tx}
H(\bX_{\bi} \mid \leaf_T(\bX)) \ge 1 - \Ex_{\ell \in T}\left[ \left(\frac{\sum_i \ell_i}{2n}\right)^2 \right].
\end{equation}
Combining~\eqref{eq:entropy-fx}--\eqref{eq:entropy-tx}, we obtain
$$
\Big(\sum_i \hat{f}(i) \Big)^2 \le 2\ln 2 \cdot 4 \mu(1-\mu) \Ex_\ell\left[\Big(\sum_i \ell_i\Big)^2\right]
$$
and the theorem follows from the bound in Lemma~\ref{lem:parityDT}.
\end{proof}

\begin{remark} 
A result that is similar to Theorem~\ref{thm:os2}, but with a slightly weaker bound, can also be 
obtained directly from Lemma~\ref{lem:parityDT} and Jensen's inequality. 
This approach gives the weaker bound $\sum_{i=1}^n \hat{f}(i) \le \sqrt{2d}$.
See Appendix~\ref{app:coarse} for the details.
\end{remark}

\subsection{Proof of Lemma~\ref{lem:parityDT}}
\label{sec:lemPDT}

The proof of Lemma~\ref{lem:parityDT} has three main components.
The first is a proof of the lemma for a class of parity decision trees that we call
\emph{(pairwise) correlation-free}.

\begin{definition}
The parity decision tree $T$ is \emph{(pairwise) correlation-free} when for every $i \neq j \in [n]$ and 
any path in the tree $T$, if $x_i \oplus x_j$ is fixed by the queries in the path, then so are $x_i$ and 
$x_j$.
\end{definition}

\begin{proposition}
\label{prop:uncor-variance}
Let $T$ be a correlation-free parity decision tree of depth $d$. Then
$
\E (\sum_i \ell_i)^2 \le d.
$
\end{proposition}

\begin{proof}
Consider any node $v$ in the parity decision tree that fixes the parity $x_i \oplus x_j$. Since
$T$ is correlation-free, every leaf below $v$ satisfies $\ell_i, \ell_j \neq 0$. In particular, 
$\Pr_{\ell \sim v}[ \ell_i \ell_j = -1] = \Pr_{\ell \sim v}[ \ell_i \ell_j = 1] = 1/2$ so $\E_{\ell \sim v} \ell_i \ell_j = 0$. And every path that reaches a leaf without fixing $x_i \oplus x_j$ does not set both
$x_i$ and $x_j$, so such a leaf $\ell$ satisfies $\ell_i \ell_j = 0$. This means that for every
$i \neq j$, $\E \ell_i \ell_j = 0$ and so
\begin{equation}
\label{eq:variance}
\E (\sum_i \ell_i)^2 = \sum_i \E(\ell_i)^2 + \sum_{i\neq j} \E \ell_i \ell_j \le d,
\end{equation}
where the final inequality uses the fact that at most $d$ coordinates can be fixed by the queries of
any path in $T$.
\end{proof}

We want to use Proposition~\ref{prop:uncor-variance} by showing that we can refine every
parity decision tree into an uncorrelated parity decision tree without increasing its depth by too much.
The following proposition formalizes this statement.

\begin{proposition}
\label{prop:uncorrelated-refine}
Let $T$ be a parity decision tree of depth $d$. Then there is a refinement $T'$ of $T$ which is
an uncorrelated parity decision tree of depth at most $2d$.
\end{proposition}

\begin{proof}
For each leaf of $T$, let $J$ be a set of disjoint pairs $(i,j)$ of coordinates such that $x_i \oplus x_j$
is fixed but neither $x_i$ nor $x_j$ have been fixed by the queries down the path to the leaf.
Refine $T$ by querying the first coordinate in each such pair. 
Once we have done this refinement at every leaf, the resulting tree is uncorrelated. To complete
the proof of the proposition, it remains to show that at most $2d$ disjoint pairs of correlated 
coordinates can occur in any path on the tree $T$.

 Let $V$ be the subspace of $\{0,1\}^n$ spanned by the (at most) $d$ queries down any fixed path in
$T$.  Let $S$ be a maximal linearly independent subset of $V$ containing only vectors of Hamming
weight $1$ or $2$. Since $V$ is a $d$-dimensional subspace, $|S| \le d$. Let $J$ be the set of coordinates that are set to $1$ in at least one vector in $S$. Then $|J| \le 2d$.
Furthermore, if $i$ is fixed
or correlated, there exists a vector $v$ of Hamming weight at most $2$ in $V$ for which $v_i = 1$.
This means that either $v \in S$ or $v$ is a linear combination of some vectors in $S$; either case
implies that $i \in J$.
\end{proof}

The third and final component of our proof of the lemma is a simple argument showing that
refining a decision tree can only increase the value of $\E (\sum_i \ell_i)^2$. 

\begin{proposition}
\label{prop:refine-variance}
Let $T'$ be any refinement of the parity decision tree $T$. Then
$$
\Ex_{\ell \in T} \Big(\sum_{i=1}^n \ell_i\Big)^2 \le \Ex_{\ell' \in T'} \Big(\sum_{i=1}^n \ell_i'\Big)^2.
$$
\end{proposition}

\begin{proof}
It suffices to establish the proof in the case where $T'$ replaces one leaf of $T$ with an extra
node. Let $v$ be the leaf in $T$ that we replace with the node with leaves $u, w$. Let $\rho$ be
the probability that a random input $x$ reaches the leaf $v$ in $T$. Then
$$
\Ex_{\ell' \in T'} \Big(\sum_{i=1}^n \ell_i'\Big)^2 - \Ex_{\ell \in T} \Big(\sum_{i=1}^n \ell_i\Big)^2 
= \rho \cdot \left( \tfrac12 (\sum_{i=1}^n u_i)^2 + \tfrac12 (\sum_{i=1}^n w_i)^2 - (\sum_{i=1}^n v_i)^2 \right).
$$
Let $S \subseteq [n]$ be the set of coordinates that are fixed by the query at the node that
replaced $v$. Then $v_i = 0$ for each $i \in S$, and $\delta := \sum_{i \in S} u_i = - \sum_{i \in S} w_i$.
Write $\gamma = \sum_i v_i$. Then
\[
\tfrac12 \Big(\sum_{i=1}^n u_i\Big)^2 + \tfrac12 \Big(\sum_{i=1}^n w_i\Big)^2 - \Big(\sum_{i=1}^n v_i\Big)^2
= \tfrac12 (\gamma + \delta)^2 + \tfrac12 (\gamma - \delta)^2 - \gamma^2
= \delta^2 \ge 0. \qedhere
\]
\end{proof}

We can now complete the proof of the lemma.

\begin{proof}[Proof of Lemma~\ref{lem:parityDT}]
Let $T'$ be the uncorrelated parity decision tree of depth at most $2d$ obtained by refining $T$,
as promised by Proposition~\ref{prop:uncorrelated-refine}. By Propositions~\ref{prop:refine-variance} 
and~\ref{prop:uncor-variance},
\[
\Ex_{\ell \in T} \Big(\sum_{i=1}^n \ell_i\Big)^2 \le \Ex_{\ell' \in T'} \Big(\sum_{i=1}^n \ell_i'\Big)^2 \le 2d. \qedhere
\]
\end{proof}

\begin{remark}
The same arguments in the proof of Lemma~\ref{lem:parityDT} can also be sharpened to show
that the expression $\E (\sum_i \ell_i)^2$ is bounded above by 2 times the average depth of the
parity decision tree $T$.
\end{remark}

\begin{remark}
When $T$ is a standard decision tree, \eqref{eq:variance} directly implies that $\E (\sum_i \ell_i)^2 \le d$.
It is natural to ask whether Lemma~\ref{lem:parityDT} can be sharpened to obtain the same bound
for parity decision trees as well. It cannot: consider the $\MAJ_3 : \{-1,1\}^3 \to \{-1,1\}$ function,
which returns the sign of $x_1 + x_2 + x_3$. One parity decision tree that computes $\MAJ_3$ 
queries $x_1 x_2$ at the root and then queries $x_1$ if $x_1 x_2 = 1$, or $x_3$ otherwise. This 
tree has depth $2$ but $\E (\sum_i \ell_i)^2 = \frac52 > 2$.
\end{remark}

\section{The recursive majority function}

Let us now see how Theorem~\ref{thm:os2} yields a lower bound on the parity decision tree
complexity of the recursive majority function.

\begin{reptheorem}{thm:3majk}
Every parity decision tree that computes $\MAJ_3^{\otimes k}$ 
has depth $\Omega(2.25^k)$. 
\end{reptheorem} 

\begin{proof}
By direct calculation, we observe that the Fourier expansion of the 
$\MAJ_3$ function is
$$
\MAJ_3(x_1,x_2,x_3) = 
\frac{1}{2}x_1 + \frac{1}{2}x_2 + \frac{1}{2}x_3 - \frac{1}{2}x_1 x_2 x_3.
$$ 
By Fact~\ref{fact:level-1-mult}, for every $k > 1$ we have
$$
\sum_{i \in [3^k]} \widehat{\MAJ_3^{\otimes k}}(i) = 
\left( \sum_{i \in [3]} \widehat{\MAJ_3}(i) \right)
\left( \sum_{j \in [3^{k-1}]} \widehat{\MAJ_3^{\otimes k-1}}(j) \right)
= \frac32
\left( \sum_{j \in [3^{k-1}]} \widehat{\MAJ_3^{\otimes k-1}}(j) \right).
$$
By induction, this identity yields
$$
\sum_{i \in [3^k]} \widehat{\MAJ_3^{\otimes k}}(i) = \left( \frac32 \right)^k.
$$
Let $d$ be the minimal depth of any parity decision tree that computes 
$\MAJ_3^{\otimes k}$. By Theorem~\ref{thm:os2}, we have $(\frac32)^k \leq \sqrt{4\ln2\,d}$
and so $d \ge \Omega\big( (\frac32)^{2k} \big) = \Omega(2.25^k)$.
\end{proof}

\section{Conclusion and open problem} 
\label{conclusion} 

We have shown that the O'Donnell--Servedio inequality generalizes to the setting of
parity decision trees.
A related conjecture of Parikshit Gopalan and Rocco Servedio posits that the
O'Donnell--Servedio inequality can also be generalized in a different direction as well, 
to the setting of Boolean functions with low \emph{Fourier degree}, where the Fourier
degree of a Boolean function $f$ is the size of the largest set $S$ such that $\hat{f}(S) \neq 0$. 

\newtheorem*{GS}{Gopalan--Servedio Conjecture~\cite{o2012open}}
\begin{GS}
Let $f\isafunc$ be a Boolean function with Fourier degree $d$. Then $\sum_{i=1}^n \hat{f}(i) \le O(\sqrt{d})$. 
\end{GS} 

While the Gopalan--Servedio conjecture and Theorem~\ref{thm:os2} both generalize the O'Donnell--Servedio inequality (as Fourier degree and parity decision tree depth are both upper bounded by regular decision tree depth), they are incomparable to each other --- the $n$-variable parity function has PDT depth $1$ and Fourier degree $n$, and conversely there are functions whose PDT depth is polynomially larger than its Fourier degree~\cite{OSTWZ14}.

\section*{Acknowledgements} 

We thank Ryan O'Donnell and Rocco Servedio for insightful conversations. 
We also thank the anonymous referees of an earlier version of this manuscript for valuable feedback.

\bibliography{pdt}{}

\newcommand{\etalchar}[1]{$^{#1}$}
\begin{thebibliography}{TWXZ13}

\bibitem[Blu03]{Blum:03tutorial}
Avrim Blum.
\newblock Machine learning: a tour through some favorite results, directions,
  and open problems.
\newblock FOCS 2003 tutorial slides, available at
  http://www-2.cs.cmu.edu/~avrim/Talks/FOCS03/tutorial.ppt, 2003.

\bibitem[BSK12]{ben2012affine}
Eli Ben-Sasson and Swastik Kopparty.
\newblock Affine dispersers from subspace polynomials.
\newblock {\em SIAM Journal on Computing}, 41(4):880--914, 2012.

\bibitem[CT91]{CoverThomas:91}
Thomas~M. Cover and Joy~A. Thomas.
\newblock {\em {Elements of Information Theory}}.
\newblock Wiley, 1991.

\bibitem[CT14]{cohen2014two}
Gil Cohen and Avishay Tal.
\newblock Two structural results for low degree polynomials and applications.
\newblock {\em arXiv preprint}, 1404.0654, 2014.

\bibitem[JKS03]{JKS03}
T.~S. Jayram, Ravi Kumar, and D.~Sivakumar.
\newblock Two applications of information complexity.
\newblock In {\em Proceedings of the 35th Annual {ACM} Symposium on Theory of
  Computing, June 9-11, 2003, San Diego, CA, {USA}}, pages 673--682, 2003.

\bibitem[Leo13]{Leo13}
Nikos Leonardos.
\newblock An improved lower bound for the randomized decision tree complexity
  of recursive majority,.
\newblock In {\em Proceedings of Automata, Languages, and Programming - 40th
  International Colloquium, {ICALP} 2013, Part {I}}, pages 696--708, 2013.

\bibitem[MNS{\etalchar{+}}13]{MNS+13}
Fr{\'{e}}d{\'{e}}ric Magniez, Ashwin Nayak, Miklos Santha, Jonah Sherman,
  G{\'{a}}bor Tardos, and David Xiao.
\newblock Improved bounds for the randomized decision tree complexity of
  recursive majority.
\newblock {\em arXiv preprint}, 1309.7565, 2013.

\bibitem[MNSX11]{MNSX11}
Fr{\'{e}}d{\'{e}}ric Magniez, Ashwin Nayak, Miklos Santha, and David Xiao.
\newblock Improved bounds for the randomized decision tree complexity of
  recursive majority.
\newblock In {\em Automata, Languages and Programming - 38th International
  Colloquium, {ICALP} 2011, Zurich, Switzerland, July 4-8, 2011, Proceedings,
  Part {I}}, pages 317--329, 2011.

\bibitem[MO09]{montanaro2009communication}
Ashley Montanaro and Tobias Osborne.
\newblock On the communication complexity of xor functions.
\newblock {\em arXiv preprint}, 0909.3392, 2009.

\bibitem[O'D12]{o2012open}
Ryan O'Donnell.
\newblock Open problems in analysis of {B}oolean functions.
\newblock {\em arXiv preprint}, 1204.6447, 2012.

\bibitem[O'D14]{o2007analysis}
Ryan O'Donnell.
\newblock {\em Analysis of {B}oolean {F}unctions}.
\newblock Cambridge University Press (available online at
  \url{http://analysisofbooleanfunctions.org}), 2014.

\bibitem[OS08]{OS08b}
Ryan O'Donnell and Rocco Servedio.
\newblock Learning monotone decision trees in polynomial time.
\newblock {\em SIAM Journal on Computing}, 37(3):827--844, 2008.

\bibitem[OST{\etalchar{+}}14]{OSTWZ14}
Ryan O'Donnell, Xiaorui Sun, Li{-}Yang Tan, John Wright, and Yu~Zhao.
\newblock A composition theorem for parity kill number.
\newblock In {\em {IEEE} 29th Conference on Computational Complexity, {CCC}
  2014, Vancouver}, pages 144--154, 2014.

\bibitem[Sha11]{shaltiel2011dispersers}
Ronen Shaltiel.
\newblock Dispersers for affine sources with sub-polynomial entropy.
\newblock In {\em Foundations of Computer Science (FOCS), 2011 IEEE 52nd Annual
  Symposium on}, pages 247--256. IEEE, 2011.

\bibitem[STV14]{shpilka2014structure}
Amir Shpilka, Avishay Tal, and Ben~Lee Volk.
\newblock On the structure of boolean functions with small spectral norm.
\newblock In {\em Proceedings of the 5th conference on Innovations in
  theoretical computer science}, pages 37--48. ACM, 2014.

\bibitem[SW86]{SW86}
Michael~E. Saks and Avi Wigderson.
\newblock Probabilistic boolean decision trees and the complexity of evaluating
  game trees.
\newblock In {\em 27th Annual Symposium on Foundations of Computer Science,
  Toronto, Canada, 27-29 October 1986}, pages 29--38, 1986.

\bibitem[TWXZ13]{tsang2013fourier}
Hing~Yin Tsang, Chung~Hoi Wong, Ning Xie, and Shengyu Zhang.
\newblock Fourier sparsity, spectral norm, and the log-rank conjecture.
\newblock In {\em 2013 IEEE 54th Annual Symposium on Foundations of Computer
  Science}, pages 658--667. IEEE, 2013.

\bibitem[ZS09]{zhang2009communication}
Zhiqiang Zhang and Yaoyun Shi.
\newblock Communication complexities of symmetric xor functions.
\newblock {\em Quantum Information \& Computation}, 9(3):255--263, 2009.

\bibitem[ZS10]{zhang2010parity}
Zhiqiang Zhang and Yaoyun Shi.
\newblock On the parity complexity measures of boolean functions.
\newblock {\em Theoretical Computer Science}, 411(26):2612--2618, 2010.

\end{thebibliography}
\bibliographystyle{alpha}

\appendix

\section{Mulitiplicativity of the level-1 Fourier mass}\label{sec:mult}

Fact~\ref{fact:level-1-mult} is a direct consequence of the following identity.

\begin{proposition}
For any function $f:\{-1,1\}^m\to \{-1,1\}$, any balanced function 
$g:\{-1,1\}^n \to \{-1,1\}$, 
and any $i \in [m]$ and $j \in [n]$,
$$
\widehat{f \circ g}\big((i-1)n + j\big) = \hat{f}(i)\hat{g}(j).
$$
\end{proposition}

\begin{proof}
By definition, the Fourier expansion of $f$, and linearity of expectation,
\begin{align}
\widehat{f \circ g}\big((i-1)n + j\big) 
&= \E_x\left[ f\big( g(x_1,\ldots,x_n),\ldots, g(x_{(m-1)n+1},\ldots,x_{mn})\big) 
             \cdot x_{(i-1)n + j}\right] \nonumber \\
\label{eq:mult-fourier-sum}
&= 
\sum_{S \subseteq [n]} \hat{f}(S) 
 \E_x\left[ \prod_{k \in S} g(x_{(k-1)n+1},\ldots,x_{kn}) \cdot x_{(i-1)n + j}\right].
\end{align}
When $i \notin S$,
$$
\E_x\left[ \prod_{k \in S} g(x_{(k-1)n+1},\ldots,x_{kn}) \cdot x_{(i-1)n+j}\right]
= \E_x\left[ \prod_{k \in S} g(x_{(k-1)n+1},\ldots,x_{kn})\right] \cdot \E_x \left[ x_{(i-1)n+j}\right] = 0.
$$
Similarly, when $S \setminus \{i\} \neq \emptyset$, we can fix any
$\ell \in S \setminus \{i\}$ and observe that
\begin{align*}
\E_x\left[ \prod_{k \in S} \right.& \left. g(x_{(k-1)n+1},\ldots,x_{kn}) \cdot x_{(i-1)n+j}\right] \\
&= \E_x\left[ g(x_{(\ell-1)n+1},\ldots,x_{\ell n}) \right] \cdot
\E_x\left[ \prod_{k \in S \setminus \{\ell\}} g(x_{(k-1)n+1},\ldots,x_{kn}) \cdot x_{(i-1)n+j}\right].
\end{align*}
When $g$ is balanced, $\E_x\left[ g(x_{(\ell-1)n+1},\ldots,x_{\ell n}) \right] = 0$
so the only non-zero term of the sum in~\eqref{eq:mult-fourier-sum} is
the one where $S = \{i\}$ and 
\begin{align*}
\widehat{f \circ g}\big((i-1)+j\big) 
&= \hat{f}(i) \E_x[ g(x_{(i-1)n+1},\ldots,x_{in}) x_{(i-1)n+j}] \\
&= \hat{f}(i) \E_x[ g(x_1,\ldots,x_n) x_j] \\
&= \hat{f}(i) \hat{g}(j). \qedhere
\end{align*}
\end{proof}

\section{Coarser bounds}
\label{app:coarse}

We can obtain a weaker version of Theorem~\ref{thm:os2} by combining 
Lemma~\ref{lem:parityDT} with the following easy inequality which is essentially 
equivalent to Lemma~3 in~\cite{OS08b}. 

\begin{lemma}[O'Donnell and Servedio~\cite{OS08b}]
\label{lem:os-easy}
Let $f : \{-1,1\}^n \to \{-1,1\}$ be computable by a parity decision $T$. Then
$$
\sum_{i=1}^n \hat{f}(i) < \E_{\ell \in T} \left[ \left|\sum_{i=1}^n \ell_i\right| \right].
$$
\end{lemma}

\begin{proof} 
The linear Fourier coefficients
of $f$ satisfy
$$
\hat{f}(i) 
= \E_x[ f(x) x_i ] 
= \E_{\ell \in T} \E_{x : t(x) = \ell}[ f(x) x_i ] 
= \E_{\ell \in T} \left[ f(\ell) \E_{x : t(x) = \ell} [x_i] \right] 
= \E_{\ell \in T} [ f(\ell) \ell_i ].
$$
So 
$
\sum_i \hat{f}(i) = \E_{\ell \in T}[ f(\ell) \sum_i \ell_i ] \le
\E_{\ell \in T}[ |\sum_i \ell_i| ].
$
\end{proof}

We are now ready to complete the proof of the slightly weaker version of 
Theorem~\ref{thm:os2}.

\begin{theorem}
\label{thm:os2-weaker}
Let $f : \{-1,1\}^n \to \{-1,1\}$ be computable by a parity decision tree of depth $d$.  Define
$\sigma^2 = 4 \Pr[ f(x) = 1 ] \Pr[ f(x) = -1 ]$ to be the variance of $f$. Then
$$ \sum_{i=1}^n \hat{f}(i) \le \sqrt{2d}.$$
\end{theorem}

\begin{proof}
By Lemma~\ref{lem:os-easy} and Jensen's inequality,
$$
\left( \sum_{i=1}^n \hat{f}(i) \right)^2 \le \E_{\ell \in T} \left[ \left|\sum_{i=1}^n \ell_i\right| \right]^2
\le 
\E_{\ell \in T} \left[ \left(\sum_{i=1}^n \ell_i\right)^2 \right].
$$
Theorem~\ref{thm:os2} then follows directly from Lemma~\ref{lem:parityDT}.
\end{proof}

\end{document}